\documentclass[11pt,a4paper]{article}
\usepackage[latin1]{inputenc}
\usepackage[compress]{cite}
\usepackage{amsmath}
\usepackage{amsfonts}
\usepackage{amssymb}
\usepackage{amsthm}
\usepackage{mathrsfs}  
\usepackage{graphicx}
\usepackage{hyperref}

\theoremstyle{plain}
\newtheorem{theorem}{Theorem}[section]

\begin{document}

%\markboth{Hugo R. C. Ferreira}{Renormalized vacuum polarization of rotating black holes}

\title{\Large{\textbf{Renormalized vacuum polarization of rotating black holes}}\footnote{Based on a talk given by H.R.C.F. at ``VII Black Holes Workshop'', Aveiro, 18-19 December 2014.}}

\author{Hugo R. C. Ferreira\footnote{Work done in collaboration with Jorma Louko.} \vspace*{2ex} \\ 
\small{School of Mathematical Sciences, University of Nottingham,} \\
\small{Nottingham NG7 2RD, United Kingdom} \\
\small{\url{pmxhrf@nottingham.ac.uk}}
}

\date{}

\maketitle

\vspace*{-3ex}

%% ABSTRACT

\begin{abstract}
Quantum field theory on rotating black hole spacetimes is plagued with technical difficulties. Here, we describe a general method to renormalize and compute the vacuum polarization of a quantum field in the Hartle-Hawking state on rotating black holes. We exemplify the technique with a massive scalar field on the warped AdS${}_3$ black hole solution to topologically massive gravity, a deformation of (2+1)-dimensional Einstein gravity. We use a ``quasi-Euclidean'' technique, which generalizes the Euclidean techniques used for static spacetimes, and we subtract the divergences by matching to a sum over mode solutions on Minkowski spacetime. This allows us, for the first time, to have a general method to compute the renormalized vacuum polarization (and, more importantly, the renormalized stress-energy tensor), for a given quantum state, on a rotating black hole, such as the physically relevant case of the Kerr black hole in four dimensions.
\end{abstract}

%\keywords{black hole; quantum field theory; renormalization.}

%\ccode{PACS numbers: 04.62.+v, 04.70.Dy}

%\tableofcontents

\vspace*{3ex}

%% INTRODUCTION

\section{Introduction}

The study of quantum field theory on stationary black hole spacetimes has proven to be much more challenging than on static black hole spacetimes. In particular, it has not been possible to compute the expectation values of the renormalized stress-energy tensor for a matter field in a given quantum state on the Kerr spacetime. Apart from the technical complexities of the calculations, other difficulties include the nonexistence of generalizations of the Hartle-Hawking state and of Euclidean methods used in static spacetimes.

Recently, Ref.~\cite{Ferreira:2014ina} laid down a method to compute the renormalized vacuum polarization of a quantum field in the Hartle-Hawking state on a rotating warped AdS${}_3$ black hole surrounded by a mirror with Dirichlet boundary conditions, implementing the Hadamard renormalization prescription on the complex Riemannian section of the spacetime. Here, I review this method and extend the results for a generic (2+1)-dimensional rotating black hole spacetime. This calculation can be taken as a warm-up for the computation of the more physically relevant renormalized stress-energy tensor. We anticipate that this technique can be extended to a wider class of rotating black hole spacetimes, in particular the Kerr black hole in four dimensions.

The contents of this note are as follows. In Sec.~\ref{section:scalar}-\ref{section:renormalization}, we outline the method for a generic (2+1)-dimensional stationary black hole spacetime and apply it in Sec.~\ref{section:numerical} to the case of the warped AdS${}_3$ black hole. We present our conclusions in Sec.~\ref{section:conclusions} and the proof of a theorem is left to Appendix~\ref{section:appendix}. Throughout this note we use the $(-,+,+)$ signature and units in which $\hbar = c = G = k_B = 1$.

%% SCALAR FIELD

\section{Scalar field on a rotating black hole spacetime}
\label{section:scalar}

For concreteness, in the following we will consider a generic (2+1)-dimensional stationary black hole spacetime, whose metric is of the form described below. This is the case of the warped AdS${}_3$ black hole described in Sec.~\ref{section:numerical}. However, as argued, the method should be applicable to a wide range of rotating black hole spacetimes in three and more spacetime dimensions.

In spherical coordinates $(t,r,\theta)$, the metric of a (2+1)-dimensional stationary black hole can be written as
\begin{equation}
ds^2 = - N^2(r) \, dt^2 + g_{rr}(r) \, dr^2 + g_{\theta\theta}(r) \left( d\theta + N^{\theta}(r) \, dt \right)^2 \, ,
\end{equation}
where $N(r)$ is the lapse function and $N^{\theta}(r)$ is the shift function. Let the event horizon be at $r = r_+$ and let $\Omega_{\mathcal{H}}$ be the angular velocity of the horizon with respect to the coordinate system. Denote by $\chi = \partial_t + \Omega_{\mathcal{H}} \, \partial_{\theta}$ the Killing vector field which generates the horizon and which is timelike for $r \in (r_+, r_{\mathcal{C}})$, where $r = r_{\mathcal{C}}$ is the location of the speed-of-light surface, if it exists.

For most cases of physical interest, such as the Kerr spacetime, there is not any timelike Killing vector field in all of the exterior region $r > r_+$. As a consequence, in the context of quantum field theory, there is not a well defined quantum vacuum state which is regular at the horizon and is invariant under the isometries of the spacetime \cite{Kay:1988mu}. However, a vacuum state with these properties can be defined if we restrict the spacetime by inserting a boundary $\mathcal{M}$ at constant radius $r = r_{\mathcal{M}}$, with $r_+ < r_{\mathcal{M}} < r_{\mathcal{C}}$, in which the field satisfies Dirichlet boundary conditions, as in this region $\chi$ is a timelike Killing vector field.

For convenience, we change to ``corotating coordinates'' $(\tilde{t}=t, \, r, \,\tilde{\theta} = \theta - \Omega_{\mathcal{H}} t)$, such that $\chi$ is given by $\chi = \partial_{\tilde{t}}$ and the metric is then given by
\begin{equation}
ds^2 = - N^2(r) \, d\tilde{t}^2 + g_{rr}(r) \, dr^2 + g_{\theta\theta}(r) \left( d\tilde{\theta} + \big( N^{\theta}(r) + \Omega_{\mathcal{H}} \big) d\tilde{t} \right)^2 \, .
\label{eq:metriccorotatingcoords}
\end{equation}

From now on, we only consider the portion of the exterior region from the horizon up to the mirror, which we denote by $\widetilde{\text{I}}$, in contrast with the full exterior region I. A real massive scalar field $\Phi$ on $\widetilde{\text{I}}$ obeys the Klein-Gordon equation
\begin{equation}
\left(\nabla^2 - m^2 - \xi R \right) \Phi = 0 \, ,
\label{eq:fieldequation1}
\end{equation}
where $m$ is the mass of the field, $R$ is the Ricci scalar and $\xi$ is the curvature coupling parameter. We consider mode solutions of \eqref{eq:fieldequation1} of the form
\begin{equation}
\Phi_{\tilde{\omega} k}(\tilde{t},r,\tilde{\theta}) = e^{-i \tilde{\omega} \tilde{t} + i k \tilde{\theta}} \, \phi_{\tilde{\omega} k}(r) \, ,
\label{eq:fieldansatz}
\end{equation}
where $\tilde{\omega} \in \mathbb{R}$ and $k \in \mathbb{Z}$.

Given the construction above, it is possible to show that there is a natural positive and negative frequency decomposition of the mode solutions with respect to the affine parameters on the horizon \cite{Ferreira:2014ina}. It is then possible to canonically quantize the scalar field and define a vacuum state $|H \rangle$, which is regular at the horizons and invariant under the spacetime isometries. Therefore, $|H \rangle$ has all of the properties of the ``Hartle-Hawking vacuum state'', as usually defined for static black hole spacetimes \cite{Hartle:1976tp}.

The Feynman propagator is defined as
\begin{equation}
G^{\text{F}}(x,x') := i \, \langle H | \mathscr{T} \left( \Phi(x) \Phi(x') \right) | H \rangle \, , 
\label{eq:feynmanpropagator}
\end{equation}
where $\mathscr{T}$ is the time-ordering operator and $\Phi(x)$ is interpreted as an operator-valued distribution which acts on the Fock space of the theory. The Feynman propagator is one of the Green's functions associated with the Klein-Gordon equation.

The Feynman propagator, evaluated for certain quantum states states, has a Hadamard expansion of the form\cite{Decanini:2005eg}
\begin{equation}
G^{\text{F}}(x,x') = \frac{i}{4 \sqrt{2} \pi} \left( \frac{U(x,x')}{\sqrt{\sigma(x,x') + i \epsilon}} + W(x,x') \right) \, , \qquad \epsilon \to 0 \! + \, .
\label{eq:GFHadamardexpansion}
\end{equation}
Here, $x$ and $x'$ belong to a neighborhood $N$ and are linked by a unique geodesic which lies entirely in $N$; $\sigma (x,x')$ is half of the squared geodesic distance between $x$ and $x'$; $U(x,x')$ and $W(x,x')$ are symmetric and regular biscalar functions. Note that this expansion is only valid for (2+1)-dimensional spacetimes. Spacetimes with $d>2$ dimensions have a singular term of the form $[\sigma(x,x') + i \epsilon]^{d/2-1}$ and, in the case of spacetimes with even dimension $d \geq 4$, an additional singular contribution of the form $\log[\sigma(x,x') + i \epsilon]$ is present.

A quantum state for which the short-distance singularity structure of $G^{\text{F}}$ is given by \eqref{eq:GFHadamardexpansion} (or its higher-dimensional versions) is called a ``Hadamard state''. Examples include the Minkowski state on the Minkowski spacetime and the Hartle-Hawking state on the Schwarszchild black hole \cite{Sanders:2013pka}. Even though, as noted above, there is not an analogue of the Hartle-Hawking state on a rotating black hole \cite{Kay:1988mu}, an isometry-invariant state which is regular at the horizons can be defined if a Dirichlet boundary is present, so that the quantum field cannot reach the speed-of-light surface. We still call this state the ``Hartle-Hawking state''.

Concerning the regular functions $U(x,x')$ and $W(x,x')$ in the expansion \eqref{eq:GFHadamardexpansion}, it can be shown (see e.g.~Ref.~\cite{Decanini:2005eg}) that $U(x,x')$ only depends on the geometry along the geodesics joining $x$ to $x'$, whereas $W(x,x')$ contains the quantum state dependence of the Feynman propagator. Therefore, the singular, state-independent part of the Feynman propagator, the ``Hadamard singular part'', is
\begin{equation}
G_{\text{Had}}(x,x') := \frac{i}{4\sqrt{2}\pi} \frac{U(x,x')}{\sqrt{\sigma(x,x')+i\epsilon}} \, .
\end{equation}

The renormalized vacuum polarization $\langle \Phi^2(x) \rangle$ in any Hadamard state is then defined to be
\begin{equation}
\langle \Phi^2(x) \rangle := - i \lim_{x' \to x} \left[ G^{\text{F}}(x,x') - G_{\text{Had}}(x,x') \right] \, .
\label{eq:phisquareddef}
\end{equation}

At this stage, the Euclidean techniques used for static spacetimes to more easily compute the Feynman propagator are not available for stationary, but not static, spacetimes. In the next section, we describe a ``quasi-Euclidean'' method to achieve the same result.

%% QUASI-EUCLIDEAN METHOD

\section{Quasi-Euclidean method}
\label{section:quasi-euclidean}

Euclidean methods can be a powerful tool to do calculations in quantum field theory on static spacetimes. A static spacetime can be thought of as a real Lorentzian section of a complex manifold, for which it is always possible to find a real Riemannian (or ``Euclidean'') section by performing an appropriate analytical continuation. For a (2+1)-dimensional static spacetime whose metric in coordinates $(t,r,\theta)$ is
\begin{equation}
ds^2 = - N^2(r) \, dt^2 + g_{rr}(r) \, dr^2 + g_{\theta\theta}(r) \, d\theta^2 \, ,
\end{equation}
where $t$ is a global timelike coordinate, one can obtain the real Riemannian section by performing a Wick rotation $t \to -i \tau$, $\tau \in \mathbb{R}$,
\begin{equation}
ds^2_{\mathbb{R}} = N^2(r) \, d\tau^2 + g_{rr}(r) \, dr^2 + g_{\theta\theta}(r) \, d\theta^2 \, .
\end{equation}

This analytic continuation procedure does not easily generalize to spacetimes that are stationary but not static. For instance, for the exterior of a rotating black hole, there might not be a globally defined timelike vector field or an analytic continuation in the coordinates that results in a real section with a positive definite metric. Both of these issues are present in Kerr\cite{Woodhouse:1977-complex}. 

If we only consider region~$\widetilde{\text{I}}$ of the (2+1)-dimensional rotating black hole spacetime, there exists an everywhere timelike Killing vector field, $\chi = \partial_{\tilde{t}}$. If we now perform a Wick rotation $\tilde{t} = -i \tau$, with $\tau \in \mathbb{R}$, the metric \eqref{eq:metriccorotatingcoords} becomes
\begin{equation}
ds^2_{\mathbb{C}} = N^2(r) \, d\tau^2 + g_{rr}(r) \, dr^2 + g_{\theta\theta} \left( d\tilde{\theta} - i \, \big( N^{\theta}(r) + \Omega_{\mathcal{H}} \big) d\tau \right)^2 \, .
\label{eq:metriccomplexRiemanniansection}
\end{equation}
This is the complex-valued metric $g^{\mathbb{C}}$ of the ``complex Riemannian'' (or ``quasi-Euclidean'') section $I^{\mathbb{C}}$ of a complex manifold, in which region $\widetilde{\text{I}}$ is a real Lorentzian section \cite{Gibbons:1976ue,Frolov:1982pi,Brown:1990di,Moretti:1999fb}. This metric is regular at the horizon if $\tau$ is periodic with period $2\pi/\kappa_+$, where $\kappa_+$ is the surface gravity of the black hole.

The Green's function $G$ associated with the Klein-Gordon equation in the complex Riemannian section satisfies the distributional equation
\begin{equation}
\left( \nabla^2 - m^2 \right) G(x,x') = - \frac{\delta^3(x,x')}{\sqrt{g(x)}} = - \frac{\delta(\tau-\tau') \delta(r-r') \delta(\tilde{\theta}-\tilde{\theta}')}{\sqrt{g(x)}} \, ,
\label{eq:GCfunctioneq}
\end{equation}
where $g(x) := |\det(g^{\mathbb{C}}_{\mu\nu})|$ and $\nabla^2 := (g^{\mathbb{C}})^{\mu\nu} \nabla_{\mu} \nabla_{\nu}$.

In contrast to the real Lorentzian section, there is a unique solution to this equation in the complex Riemannian section which satisfies the following boundary conditions: (i) $G(x,x')$ is regular at $r = r_+$, and (ii) $G(x,x')$ satisfies the Dirichlet boundary conditions at $r = r_{\mathcal{M}}$. This is because two of the directions of the complex spacetime are periodic, while the third direction is compact. In contrast, on static spacetimes without any boundary (and suitable asymptotic properties at infinity), the Euclidean section has a unique Euclidean Green's function, due to the ellipticity of the Klein-Gordon operator.

Given the periodicity conditions of $\tau$ and $\tilde{\theta}$, one has
\begin{align}
\delta (\tau - \tau') &= \frac{\kappa_+}{2\pi} \sum_{n = -\infty}^{\infty} e^{i \kappa_+ n (\tau - \tau')} \, ,
\label{eq:deltatau} \\
\delta (\tilde{\theta} - \tilde{\theta}') &= \frac{1}{2\pi} \sum_{k = -\infty}^{\infty} e^{i k(\tilde{\theta} - \tilde{\theta}')} \, ,
\label{eq:deltatheta}
\end{align}
understood as distributional identities. We can than find the unique solution of Eq.~\eqref{eq:GCfunctioneq} and write it as a mode sum,
\begin{equation}
G(x,x') = \frac{\kappa_+}{4\pi^2} \sum_{n = -\infty}^{\infty} e^{i \kappa_+ n (\tau - \tau')} \sum_{k = -\infty}^{\infty} e^{i k(\tilde{\theta} - \tilde{\theta}')} \, G_{nk}(r,r')
=: \sum_{n,k} G^{\text{BH}}_{nk}(x,x') \, .
\label{eq:Greenfunction0}
\end{equation}

In practice, the radial part of the Green's function $G_{nk}(r,r')$ is found by solving the radial field equation, using standard Green's function techniques. In 2+1 dimensions it is generally possible to write it in closed form, whereas we need to resort to numerical methods for four or more dimensions.

%% HADAMARD RENORMALIZATION

\section{Hadamard renormalization}
\label{section:renormalization}

As with with the Feynman propagator in the real Lorentz section, the Green's function $G$ in the complex Riemannian section has an Hadamard expansion and, in particular, we can write its Hadamard singular part as
\begin{equation}
G_{\text{Had}}(x,x') = \frac{1}{4\sqrt{2}\pi} \frac{U(x,x')}{\sqrt{\sigma(x,x')}} \, ,
\label{eq:Hadamardsingpart}
\end{equation}
where $\sigma(x,x')$ is the generalization of the quantity defined in the real Lorentz section (see Ref.~\cite{Ferreira:2014ina} for more details).

In an analogous way to the Lorentzian case, we now subtract the Hadamard singular part from the Green's function $G$ and obtain the vacuum polarization,
\begin{equation}
\langle \Phi^2(x) \rangle = \lim_{x' \to x} \left[ G(x,x') - G_{\text{Had}}(x,x') \right] \, .
\label{eq:vacuumpolCsection}
\end{equation}
By construction, the Green's function $G$ is regular at $r = r_+$, satisfies the Dirichlet boundary conditions at $r = r_{\mathcal{M}}$ and is invariant under the spacetime isometries. Therefore, $\langle \Phi^2(x) \rangle$ as given by \eqref{eq:vacuumpolCsection} is the vacuum polarization for a scalar field in the Hartle-Hawking state.

It remains to perform the subtraction in \eqref{eq:vacuumpolCsection} before the coincidence limit can be taken. As $G$ is known only as the mode sum \eqref{eq:Greenfunction0}, the evaluation of $\langle \Phi^2(x) \rangle$ requires $G_{\text{Had}}$ to be rewritten as a mode sum that can be combined with \eqref{eq:Greenfunction0} so that the divergences in the coincidence limit get subtracted under the sum term by term. We accomplish this by comparing $G_{\text{Had}}$ to the Hadamard singular part for a scalar field in rotating Minkowski spacetime in the complex Riemannian section, and by rewriting it as a sum over Minkowski mode solutions. 

It can be shown that the Hadamard singular part of the Green's function for a scalar field in the Minkowski vacuum can be written as 
\begin{equation}
G_{\text{Had}}^{\mathbb{M}}(x,x') = \sum_{n,k} G^{\mathbb{M}}_{nk}(x,x') + G_{\text{reg}}^{\mathbb{M}}(x,x') \, ,
\label{eq:GnkMink}
\end{equation}
where the first term on the rhs is the full Green's function and $G_{\text{reg}}^{\mathbb{M}}(x,x')$ is a term which is finite when $x' \to x$ \cite{Ferreira:2014ina}.

It is convenient at this stage to consider a particular choice of point separation. Assume that the black hole metric is given in coordinates $(\tau,r,\tilde{\theta})$, whereas the Minkowski metric is given in coordinates $(\tau,\rho,\tilde{\theta})$. Now, consider the case of angular separation in each spacetime, such that for the black hole case $x = (\tau, r, 0)$ and $x' = (\tau, r, \tilde{\theta})$, with $\tilde{\theta} > 0$, and similarly for the Minkowski case.

The expansion of the Hadamard singular parts for small $\tilde{\theta}$ are
\begin{align}
G_{\text{Had}}(x,x') &= \frac{1}{4\pi} \frac{1}{\sqrt{g_{\tilde{\theta}\tilde{\theta}}(r)}} \, \frac{1}{\tilde{\theta}} + \mathcal{O}(\tilde{\theta}) \, , \\
G_{\text{Had}}^{\mathbb{M}}(x,x') &= \frac{1}{4\pi} \frac{1}{\sqrt{g^{\mathbb{M}}_{\tilde{\theta}\tilde{\theta}}(\rho)}} \, \frac{1}{\tilde{\theta}} + \mathcal{O}(\tilde{\theta}) \, ,
\end{align}
where $g^{\mathbb{M}}_{\tilde{\theta}\tilde{\theta}}(\rho) = \rho^2$ is the $\tilde{\theta}\tilde{\theta}$-component of the metric for the rotating Minkowski spacetime. 
We are free to identify
\begin{equation}
g^{\mathbb{M}}_{\tilde{\theta}\tilde{\theta}}(\rho) \equiv \gamma^{-2}(r) \, g_{\tilde{\theta}\tilde{\theta}}(r) \, ,
\end{equation}
where $\gamma(r) > 0$ is a function to be specified. This identification provides a matching between the two radial coordinates, $\rho = \rho(r) = \gamma^{-1}(r) \sqrt{g_{\tilde{\theta}\tilde{\theta}}(r)}$.

Given this identification, we can now write
\begin{align}
G(x,x') - G_{\text{Had}}(x,x') = \sum_{n,k} \left[ G^{\text{BH}}_{nk}(x,x') - \gamma^{-1} G^{\mathbb{M}}_{nk}(x,x') \right] - \gamma^{-1} G_{\text{reg}}^{\mathbb{M}}(x,x') + \mathcal{O}(\tilde{\theta}) \, .
\label{eq:Gdiff}
\end{align}

The Minkowski Green's function has several free parameters: $\rho$ (radial coordinate), $T_{\mathbb{M}}$ (temperature of the scalar field), $\Omega_{\mathbb{M}}$ (angular velocity of the coordinate system) and $m_{\mathbb{M}}^2$ (squared mass of the scalar field), besides the unspecified factor $\gamma$ we introduced above. These can be chosen such that the double sum in \eqref{eq:Gdiff} is convergent when $\tilde{\theta} \to 0$.

The main result of this section is the following:
\begin{theorem} \label{theorem}
If the parameters $\gamma$, $T_{\mathbb{M}}$ and $\Omega_{\mathbb{M}}$ are chosen as
\begin{equation}
\gamma(r) = N(r) \, , \qquad T_{\mathbb{M}} = \frac{\kappa_+}{2\pi} \, , \qquad \Omega_{\mathbb{M}} = N^{\theta}(r) + \Omega_{\mathcal{H}} \, ,
\label{eq:matching}
\end{equation}
then the double sum in \eqref{eq:Gdiff} is finite in the coincidence limit.
\end{theorem}
\begin{proof}
See Appendix~\ref{section:appendix}.
\end{proof}

This choice corresponds to have the temperature $T_{\mathbb{M}}$ of the scalar field in Minkowski to match the Hawking temperature of the black hole and to have the angular velocity $\Omega_{\mathbb{M}}$ to be equal to the one measured by a locally nonrotating observer at radius $r$ in the black hole spacetime.

The key aspect of the proof of this theorem is that, in order to remove the divergences, we only need to know the asymptotic behaviour of the Green's function summands $G^{\text{BH}}_{nk}(r,r)$ and $G^{\mathbb{M}}_{nk}(\rho,\rho)$ for large values of the quantum number $n$ and $k$, and not the full solutions. This implies that, apart from technical difficulties, this method can be applied to black holes in four or more dimensions, for which although we can only obtain the Green's functions numerically, the asymptotic expansions of the summands for large quantum numbers can be explicitly computed using the procedure described in Appendix~\ref{section:appendix} (and in Ref.~\cite{Ferreira:2014ina} in more detail).

Setting the parameters as in \eqref{eq:matching}, it is now possible to take the coincidence limit $\tilde{\theta} \to 0$ of \eqref{eq:Gdiff} and compute the renormalized vacuum polarization \eqref{eq:vacuumpolCsection}. In the next section, as an example, we present the numerical results for the particular case of the warped AdS${}_3$ black hole.

%% NUMERICAL RESULTS

\section{Numerical results for the warped AdS${}_3$ black hole}	
\label{section:numerical}

In this section, we exemplify the method described above by numerically computing the vacuum polarization of the scalar field in the Hartle-Hawking state for a (2+1)-dimensional rotating black hole, the spacelike stretched black hole. This is one of the several types of warped AdS${}_3$ black hole solutions \cite{Anninos:2008fx} to topologically massive gravity \cite{Deser:1982vy,Deser:1981wh}, an extension of Einstein gravity in 2+1 dimensions with a propagating degree of freedom. Its metric, in coordinates $(t,r,\theta)$, is given by\footnote{We set the cosmological length $\ell = 1$.}
\begin{equation}
ds^2 = - N^2(r) dt^2 + \frac{dr^2}{4 R^2(r) N^2(r)} + R^2(r) \left( d\theta + N^{\theta}(r) dt \right)^2 \, ,
\label{eq:metricbh}
\end{equation}
with
\begin{subequations}
\begin{align}
R^2(r) &= \frac{r}{4} \left[ 3(\nu^2-1)r + (\nu^2+3)(r_+ + r_-) - 4\nu \sqrt{r_+ r_-(\nu^2+3)} \right] \, , \\
N^2(r) &= \frac{(\nu^2+3)(r-r_+)(r-r_-)}{4R^2(r)} \, , \\
N^{\theta}(r) &= \frac{2\nu r - \sqrt{r_+ r_- (\nu^2+3)}}{2 R^2(r)} \, .
\end{align}
\end{subequations}

There are outer and inner horizons at $r = r_+$ and $r = r_-$, respectively, where the coordinates $(t, r, \theta)$ become singular, and a singularity at $r=r_0$ (the largest zero of $R^2(r)$). The dimensionless coupling $\nu \in (1, \infty)$ is the warp factor, and in the limit $\nu \to 1$ the above metric reduces to the metric of the BTZ black hole\cite{Banados:1992wn,Banados:1992gq} in a rotating frame. For $\nu > 1$, there exists a speed-of-light surface is located at
\begin{equation}
r = r_{\mathcal{C}} = \frac{4 \nu^2 r_+ - (\nu^2+3) r_-}{3(\nu^2-1)} \, .
\end{equation}
We assume that a mirror $\mathcal{M}$ is present at constant radius $r = r_{\mathcal{M}} < r_{\mathcal{C}}$, at which Dirichlet boundary conditions are imposed. Moreover, the angular velocity of the outer horizon is given by
\begin{equation}
\Omega_{\mathcal{H}} = - \frac{2}{2\nu r_+ - \sqrt{r_+ r_- (\nu^2+3)}} \, .
\end{equation}
Note in particular that this spacetime does not have a static limit.

The Carter-Penrose diagram for this spacetime when $r_0 < r_- < r_+$ is shown in Fig.~\ref{fig:CPdiagrams}, which is essentially of the same form of those of asymptotically flat black holes in 3+1 dimensions. More details about this black hole solution can be found in Ref.~\cite{Ferreira:2013zta} and references therein. 

\begin{figure}[t!]
\begin{center}
\includegraphics[scale=0.5]{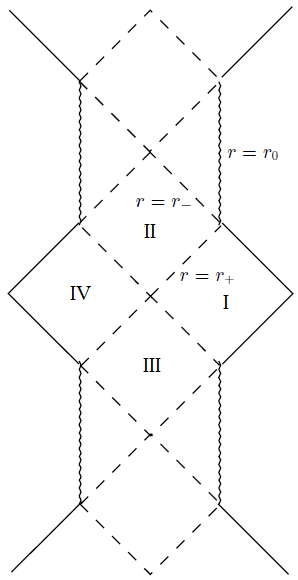} 
\hspace*{10ex}
\includegraphics[scale=0.5]{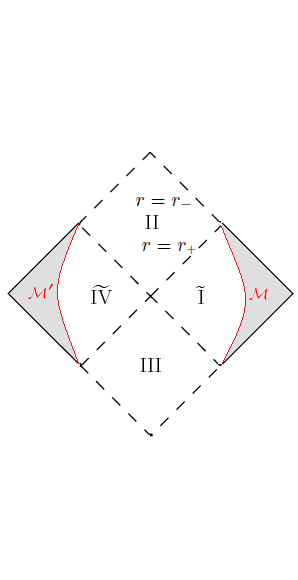}
\caption{\label{fig:CPdiagrams} Carter-Penrose diagrams of the spacelike stretched black hole spacetime for the case $r_0 < r_- < r_+$ on the left and a portion of the spacetime with the mirrors in place on the right.}
\end{center}
\end{figure}

As described in Sec.~\ref{section:scalar}, we restrict our attention only to region $\widetilde{\text{I}}$ of the spacetime (see Fig.~\ref{fig:CPdiagrams}). In this case, the vacuum polarization of the scalar field in the Hartle-Hawking state is given by
\begin{align}
\langle \Phi^2 (x) \rangle &= \sum_{k=-\infty}^{\infty} \sum_{n=-\infty}^{\infty} \left[ G^{\text{BH}}_{nk}(r,r) - \frac{1}{N(r)} \, G^{\mathbb{M}}_{nk}\left(\tfrac{R(r)}{N(r)},\tfrac{R(r)}{N(r)}\right)\Big|_{\Omega_{\mathbb{M}} = N^{\theta}(r)+\Omega_{\mathcal{H}}} \right] \notag \\ 
&\quad + \frac{1}{4\pi N(r)} \left[ {-m_{\mathbb{M}}} + \sum_{N \neq 0} \frac{e^{-m_{\mathbb{M}} \sqrt{\left(\frac{N}{T_{\mathbb{M}}}\right)^2 - \tfrac{4 R^2(r)}{N^2(r)} \sinh^2 \left( \frac{\Omega_{\mathbb{M}} N}{2T_{\mathbb{M}}} \right) + i \epsilon \, \text{sgn}(\Omega_{\mathbb{M}} N)}}}{\sqrt{\left(\frac{N}{T_{\mathbb{M}}}\right)^2 - \tfrac{4 R^2(r)}{N^2(r)} \sinh^2 \! \left( \frac{\Omega_{\mathbb{M}} N}{2T_{\mathbb{M}}} \right) + i \epsilon \, \text{sgn}(\Omega_{\mathbb{M}} N)}} \right] \, ,
\label{eq:vacuumpolcalculation}
\end{align}
with $\epsilon \to 0+$ indicating the choice of branch of the square root \cite{Ferreira:2014ina}.

The numerical results for selected values of the parameters are presented in Fig.~\ref{fig:numerics1}. In the plot, $\langle \Phi^2(x) \rangle$ is shown as a function of the normalized radial coordinate $z/z_{\mathcal{M}}$, where $z = (r-r_+)/(r-r_-)$. The plot is very similar to the one obtained in Ref.~\cite{Duffy:2002ss} for a scalar field in the (3+1)-dimensional Minkowski spacetime surrounded by a mirror with Dirichlet boundary conditions (note that ``rotating Minkowski spacetime'' is related to ``static Minkowski spacetime'' by a coordinate transformation, hence the results for $\langle \Phi^2(x) \rangle$ are the same for both cases).

Furthermore, note that $\langle \Phi^2(x) \rangle$ gets arbitrarily large and negative as the mirror is approached. This is to be expected, as we imposed that the Green's function $G(x,x')$ must vanish, even when $x' \to x$, whereas the subtraction term still diverges when $x' \to x$ (see chapter 4.3 of \cite{birrell1984quantum} for more details). 

% PLOT with numerical results
%
\begin{figure}[ht!]
\begin{center}
\includegraphics[scale=1]{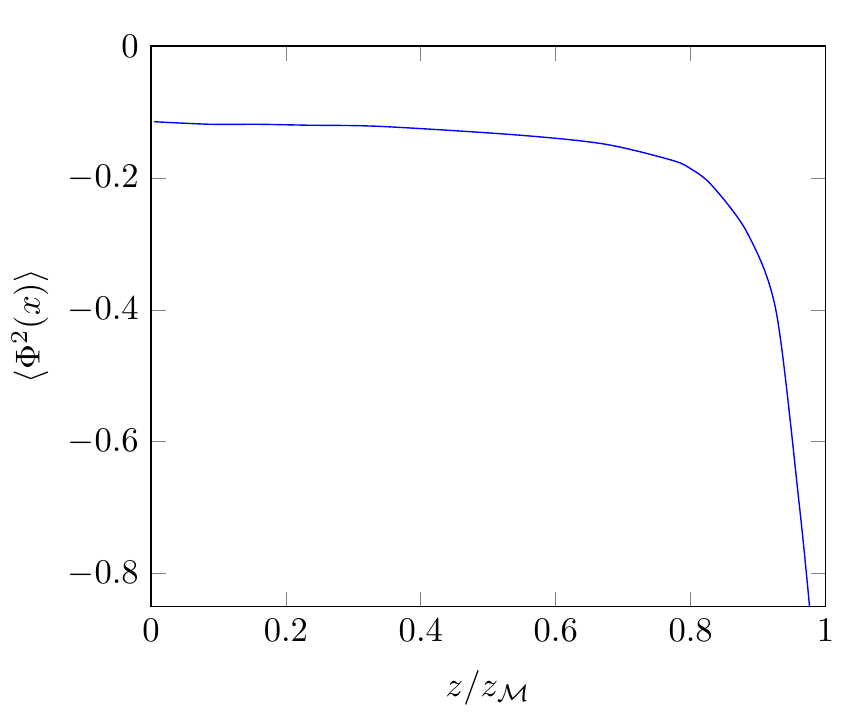} 
\caption{\label{fig:numerics1} Vacuum polarization for the scalar field as a function of the normalized radial coordinate $z/z_{\mathcal{M}}$, for $\nu = 1.2$, $r_+ = 15$, $r_- = 1$, $r_{\mathcal{M}} = 62$ and $m = 1$.}
\end{center}
\end{figure}
%

%% CONCLUSIONS

\section{Conclusions}	
\label{section:conclusions}

In this note, we have described a method to compute the vacuum polarization for a quantum field in the Hartle-Hawking state on a general rotating black hole surrounded by a mirror. We have employed a ``quasi-Euclidean'' technique to obtain the complex Riemannian section of the original spacetime, in which the unique Green's function associated with the Klein-Gordon equation can be found. This Green's function is given as a mode sum and its singular behavior in the coincidence limit can be subtracted by a sum over Minkowski modes with the same singularity structure. Taking the coincidence limit of this subtraction gives precisely the renormalized vacuum polarization. We exemplified the technique with a massive scalar field on a (2+1)-dimensional rotating spacelike stretched black hole.

A key ingredient in our implementation of the Hadamard renormalization was to match the mode sum for the Green's function in the complex Riemannian section of the black hole to a mode sum in the complex Riemannian section of a rotating Minkowski spacetime. We anticipate that, despite the technical complexity of the calculations, this method can be applied to the Kerr spacetime in four dimensions. In this case, the relevant mode solutions to the Klein-Gordon equation on the complex Riemannian section would need to be obtained numerically, but the asymptotic properties of the solutions for large values of the quantum numbers should be within analytic reach, and it is only these asymptotic properties that are required in the matching to the Minkowski mode solutions. It should prove interesting to attempt the implementation of our method for Kerr in practice.

%% ACKNOWLEDGEMENTS

\section*{Acknowledgments}

I thank Jorma Louko, who was a collaborator in this project, for reading a draft of this note. I am grateful to the organizers of the ``VII Black Holes Workshop'' for the hospitality during the workshop. I acknowledge financial support from Funda\c{c}\~{a}o para a Ci\^{e}ncia e Tecnologia (FCT)-Portugal through Grant No.\ SFRH/BD/69178/2010.

%% APPENDIX

\appendix

\section[Proof of Theorem 4.1]{Proof of Theorem \ref{theorem}}
\label{section:appendix}

The proof below is a generalization of the proof given in Ref.~\cite{Ferreira:2014ina} for the case of the warped AdS${}_3$ black hole discussed in Sec.~\ref{section:numerical}.
 
In order to check the convergence of the double sum \eqref{eq:Gdiff} in the coincidence limit, we obtain the asymptotic behaviour of the summand for large values of the quantum numbers $n$ and $k$.
  
For a black hole spacetime with metric \eqref{eq:metriccomplexRiemanniansection} in the complex Riemannian section, the Klein-Gordon equation
\begin{equation}
\left(\nabla^2 - m^2 \right) \Phi(\tau, r, \tilde{\theta}) = 0 \, ,
\label{eq:KGeqCsection}
\end{equation}
together with the ansatz $\Phi_{n k}(\tau,r,\tilde{\theta}) = e^{i \kappa_+ n \tau + i k \tilde{\theta}} \, \phi_{n k}(r)$, lead to
\begin{equation}
\left[ \frac{1}{\sqrt{g}} \frac{d}{dr} \left( \sqrt{g} \, g^{rr} \frac{d}{dr} \right) - \frac{\left( \kappa_+ n + i k \big( N^{\theta} + \Omega_{\mathcal{H}} \big) \right)^2}{N^2} - \frac{k^2}{g_{\tilde{\theta}\tilde{\theta}}} - m^2 \right] \phi_{n k} = 0 \, .
\label{eq:radialfieldeq}
\end{equation}

Define a new radial coordinate $\xi$ such that the \eqref{eq:radialfieldeq} can be written in the form
\begin{equation}
\frac{d^2 \phi_{n k}(\xi)}{d\xi^2} - \left( \chi_{nk}^2 (\xi) + \eta^2 (\xi) \right) \phi_{n k}(\xi) = 0 \, ,
\end{equation}
where $\chi_{nk}^2 (\xi)$ contains all the $n$ and $k$ dependence and is large whenever $n^2+k^2$ is large. From \eqref{eq:radialfieldeq} we obtain
\begin{equation}
\frac{d}{d\xi} = \sqrt{g} \, g^{rr} \frac{d}{dr}
\end{equation}
and
\begin{equation}
\chi_{nk}^2 = g_{\tilde{\theta}\tilde{\theta}} \left( \kappa_+ n + i k \big( N^{\theta} + \Omega_{\mathcal{H}} \big) \right)^2 + N^2 k^2 \, , \qquad
\eta^2 = g_{\tilde{\theta}\tilde{\theta}} N^2 m^2 \, .
\end{equation}

Using the results of Appendix D of Ref.~\cite{Ferreira:2014ina}, the summand $G^{\text{BH}}_{nk}(x,x)$ of the Green's function $G(x,x')$ in \eqref{eq:Greenfunction0} has the following asymptotic expansion for large $\chi_{nk}$
\begin{equation}
G^{\text{BH}}_{nk}(x,x') = \frac{\kappa_+}{4\pi^2} \frac{1}{2\chi_{nk}} + \mathcal{O}\left(\chi_{nk}^{-3}\right) \, .
\end{equation}
Analogously, for the Minkowski Green's function, the summand $G^{\mathbb{M}}_{nk}(x,x)$ in \eqref{eq:GnkMink} has the asymptotic expansion
\begin{equation}
G^{\mathbb{M}}_{nk}(x,x') = \frac{T_{\mathbb{M}}}{2\pi} \frac{1}{2\chi^{\mathbb{M}}_{nk}} + \mathcal{O}\left((\chi^{\mathbb{M}}_{nk})^{-3}\right) \, ,
\end{equation}
where
\begin{equation}
\left(\chi^{\mathbb{M}}_{nk}\right)^2(\rho) = \rho^2 \left( 2\pi T_{\mathbb{M}} n + i k \Omega_{\mathbb{M}} \right)^2 + k^2 \, .
\end{equation}

The double sum in \eqref{eq:Gdiff} will be finite in the coincidence limit if the leading term in the asymptotic expansion of the summand vanishes, that is, if the term of order $\chi_{nk}^{-1}$ of the expansion of $G^{\text{BH}}_{nk}(x,x)$ cancels with the term of order $\left(\chi^{\mathbb{M}}_{nk}\right)^{-1}$ of the expansion of $\gamma^{-1}(r) \, G^{\mathbb{M}}_{nk}(x,x)$. This only occurs if the free parameters $\gamma$, $T_{\mathbb{M}}$ and $\Omega_{\mathbb{M}}$ are chosen as
\begin{equation}
\gamma(r) = N(r) \, , \qquad T_{\mathbb{M}} = \frac{\kappa_+}{2\pi} \, , \qquad \Omega_{\mathbb{M}} = N^{\theta}(r) + \Omega_{\mathcal{H}} \, .
\end{equation}

To show that the double sum is indeed finite in the coincidence limit, we need to check that the double sum of the remaining terms in the asymptotic expansion of the summand, which are $\mathcal{O}\left(\chi_{nk}^{-3}\right)$, is finite. This can be done in the exact same way as in the proof of Theorem III.1 of Ref.~\cite{Ferreira:2014ina}.

%% BIBLIOGRAPHY

\bibliographystyle{ws-ijmpd}
\bibliography{sample}

\end{document}